\documentclass{llncs}
\usepackage{hyperref}
\usepackage{algorithm, algorithmic}
\usepackage{amsmath, amssymb, color}
\usepackage{enumerate, graphicx, hyperref}
\usepackage{cite}
\usepackage{wrapfig}
\graphicspath{{figures/}}
\usepackage{mathptm}
\usepackage{microtype}
\usepackage{todonotes}

\let\doendproof\endproof
\renewcommand\endproof{~\hfill\qed\doendproof}

\newtheorem{observation}[theorem]{Observation}

\title{Linear-time Algorithms for\\ Proportional Apportionment}
\author{Zhanpeng Cheng \and David Eppstein}
\institute{Department of Computer Science, University of California, Irvine, USA}

\begin{document}
\maketitle

\pagestyle{plain} 

\begin{abstract}
The apportionment problem deals with the fair distribution of a discrete set of $k$ indivisible resources (such as legislative seats) to $n$ entities (such as parties or geographic subdivisions). Highest averages methods are a frequently used class of methods for solving this problem. We present an $O(n)$-time algorithm for performing apportionment under a large class of highest averages methods. Our algorithm works for all highest averages methods used in practice.
\end{abstract}

\section{Introduction}
After an election, in parliamentary systems based on party-list proportional representation, the problem arises of allocating seats to parties so that each party's number of seats is (approximately) proportional to its number of votes~\cite{Lijphart-86}. Several methods, which we survey in more detail below, have been devised for calculating how many seats to allocate to each party. Often, these methods involve \emph{sequential allocation} of seats under a system of priorities calculated from votes and already-allocated seats. For instance, the Sainte-Lagu\"e method, used for elections in many countries, allocates seats to parties one at a time, at each step choosing the party that has the maximum ratio of votes to the denominator $2s+1$, where $s$ is the number of seats already allocated to the same party.

\emph{Legislative apportionment}, although mathematically resembling seat allocation, occurs at a different stage of the political system, both in parliamentary systems and in the U.S. Congress~\cite{Athanasopoulos-93,Owens-21}. It concerns using population counts to determine how many legislative seats to allocate to each state, province, or other administrative or geographic subdivision, prior to holding an election to fill those seats. Again, many apportionment methods have been developed, some closely related to seat allocation methods. For instance, a method that generates the same results as Sainte-Lagu\"e (calculated by a different formula) was proposed by Daniel Webster for congressional seat apportionment. However, although similar in broad principle, seat allocation and apportionment tend to differ in detail because of the requirement in the apportionment problem that every administrative subdivision have at least one representative. In contrast, in seat allocation, sufficiently small parties might fail to win any seats and indeed some seat allocation methods use artificially high thresholds to reduce the number of represented parties.

We may formalize
these problems mathematically as a form of \emph{diophantine approximation}:
we are given a set of $k$ indivisible resources (legislative seats) to be distributed to $n$ entities (parties or administrative subdivisions), each with score $v_i$ (its vote total or population), so that the number of resources received by an entity is approximately proportional to its score. The key constraint here is that the entities can only receive an integral amount of resources; otherwise, giving the $i$th entity $k v_i / \sum v_i$ units of resource solves the problem optimally. 
Outside of political science, forms of the apportionment problem also appear in statistics in the problem of rounding percentages in a table so that they sum to 100\% \cite{Athanasopoulos-roundper-94} and in manpower planning to allocate personnel \cite{Mayberry-78}.

Broadly, most apportionment methods can be broken down into two classes: \emph{largest remainder methods} in which a fractional solution to the apportionment problem is rounded down to an integer solution, and then the remaining seats are apportioned according to the distance of the fractional solution from the integer solution, and \emph{highest averages methods} like the Sainte-Lagu\"e method described above, in which seats are assigned sequentially prioritized by a combination of their scores and already-assigned seats. Largest remainder methods are trivial from the algorithmic point of view, but are susceptible to certain electoral paradoxes. Highest averages methods avoid this problem, and are more easily modified to fit different electoral circumstances, but appear a priori to be slower. When implemented naively, they might take as much as $O(n)$ time per seat, or $O(nk)$ overall. Priority queues can generally be used to reduce this naive bound to $O(\log n)$ time per seat, or $O(k\log n)$ overall~\cite{Campbell-ADM-07}, but this is still suboptimal, especially when there are many more seats than parties $(k\gg n)$. We show here that many of these methods can be implemented in time $O(n)$, an optimal time bound as it matches the input size. We do not expect this speedup to have much effect in actual elections, as the time to compute results is typically minuscule relative to the time and effort of conducting an election; however, the speedup we provide may be of benefit in simulations, where a large number of simulated apportionment problems may need to be solved in order to test different variations in the parameters of the election system or a sufficiently large sample of projected election outcomes.

\subsection{Highest averages}
We briefly survey here highest averages methods (or Huntington methods), a class of methods used to solve the apportionment problem~\cite{Huntington-28,BalinYoung-hunt-77}. Balinski and Young~\cite{BalinYoung-book-82} showed that divisor methods, a subclass of the highest averages methods, are the only apportionment methods that avoid undesirable outcomes such as the \emph{Alabama paradox}, in which increasing the number of seats to be allocated can cause a party's individual allocation to decrease.
Because they avoid problematic outcomes such as this one, almost all apportionment methods in use are highest averages methods.

In a highest averages method, a sequence of divisors $d_0, d_1, \dots$ is given as part of the description of the method and determines the method. To apportion the resources, each entity is assigned an initial priority $v_i / d_0$. The entity with the highest priority is then given one unit of resource and has its priority updated to use the next divisor in the sequence (i.e., if the winning entity $i$ is currently on divisor $d_j$, then its priority is updated as $v_i / d_{j+1}$). This process repeats until all the resources have been exhausted. The priorities $v_i / d_j$ are also called \emph{averages}, giving the method its name. \autoref{tab:divisors} gives the sequence of divisors for several common highest averages methods.

\begin{table}
\centering
\begin{tabular}{|l | l | l |}
  \hline
  \textbf{Method}    & \textbf{Other Names}              & \textbf{Divisors}  \\
  \hline
  Adams              & Smallest divisors                 & 0, 1, 2, 3, $\dots, j, \dots$ \\
  \hline
  Jefferson          & Greatest divisors, d'Hondt        & 1, 2, 3, 4, $\dots, j + 1, \dots$ \\
  \hline
  Sainte-Lagu\"e     & Webster, Major fractions          & 1, 3, 5, 7, $\dots, 2j + 1, \dots$ \\
  \hline
  Modified Sainte-Lagu\"e  & ---                         & 1.4, 3, 5, 7, $\dots$ \\
  \hline
  Huntington--Hill    & Equal proportions, Geometric mean & 0, \!$\sqrt{2}$, \!$\sqrt{6}$, \!$\dots, \sqrt{j (j+1)}, \dots$ \\
  \hline
  Dean               & Harmonic mean                     & 0, $4/3$, $12/5$, $\dots, \frac{2a(a+1)}{2a+1}, \dots$ \\
  \hline
  Imperiali   & ---                               & 2, 3, 4, 5, $\dots, j + 2, \dots$ \\
  \hline
  Danish      & ---                               & 1, 4, 7, 10, $\dots, 3j + 1, \dots$ \\
  \hline
\end{tabular}
\caption{Divisors for common highest averages methods.}
\label{tab:divisors}
\vspace{-5ex}
\end{table}

A zero at the start of the divisor sequence prioritizes the first assignment to each entity over any subsequent assignment, in order to ensure that (if possible) every entity is assigned at least one unit. If a zero is given, but the number of units is less than the number of entities, the entities are prioritized by their $v_i$ values.

\subsection{New results}
In this paper, we present an $O(n)$-time algorithm for simulating a highest averages method. Our algorithm works only for divisor sequences that are close to arithmetic progressions; however, this includes all methods used in practice, since this property is necessary to achieve approximately-proportional apportionment.
For divisor sequences that are already arithmetic progressions, our algorithm transforms the problem into finding the $k$th smallest value in the disjoint union of $n$ implicitly defined arithmetic progressions, which we solve in $O(n)$ time. For methods with divisor sequences close to but not equal to arithmetic progressions, we use an arithmetic progression to approximate the divisor sequence, and show that this still gives us the desired result.

\subsection{Related work}
An alternative view of a number of highest averages methods is to find a multiplier $\lambda > 0$ such that $\sum_i \left[ \lambda v_i \right] = k$, where $[\cdot]$ is a suitable rounding function for the method. The $i$th entity is then apportioned the amount $\left[ \lambda v_i \right]$. For example, the standard rounding function gives rise to the Sainte-Lagu\"e method and the floor function gives rise to the Adam method. For these methods, the problem can be solved in $O(n^2)$ time, or $O(n\log n)$ with a priority queue \cite{Dorfleitner-97, Zachariasen-05} 
\footnote{An anonymous reviewer suggested that linear-time algorithms were given previously in two Japanese papers \cite{Ito-EATCS-04, Ito-06}. However, we were unable to track down these papers nor could we determine whether their time was linear in the number of votes, seats, or parties. The second reference, in particular, does not appear to be on the IEICE website, neither searching by year and page number nor with broad search terms.}.
The algorithm works by initializing $\lambda = k / \sum_j v_j$ and iteratively choosing a new apportionment that reduces the difference between $\sum_i \left[ \lambda v_i \right]$ and $k$. The number of new apportionments can be shown to be at most $n$.

Selecting the $k$th smallest element in certain other types of implicitly defined sets has also been well studied. Gagil and Megiddo studied the assignment of $k$ workers to $n$ jobs, where the implicitly defined sets are induced by concave functions giving the utility of assigning $k_i$ workers to job $i$~\cite{GalilMegiddo-79}. Their $O(n \log^2 k)$ algorithm was improved by Frederickson and Johnson to $O(n + p \log(k/p))$ where $p = \min(k, n)$~\cite{FredericksonJohnson-82}. For implicit sets given as an $n \times m$ matrix with sorted rows and columns, Federickson and Johnson found an $O(h \log (2k/h^2))$ time algorithm, where $h = \min(\sqrt{k}, m)$ and $m \le n$ \cite{FredericksonJohnson-84}. Sorting the inputs would turn our problem  into sorted matrix selection, but the $O(n \log n)$ sorting time would already exceed our time bound.

\section{Preliminaries}

Given strictly increasing divisors $d_0, d_1, d_2, \dots$, our goal is to simulate the highest averages method induced by those divisors in time linear in the number of entities. Instead of directly selecting the entities with the $k$ largest priorities, we take advantage of the arithmetic progression structure of the divisors and consider the problem as selecting the $k$ smallest inverted priorities.  Associate the $i$th entity to the increasing sequence
\[
  A_i = \left\{ \frac{d_j}{v_i} : j = 0, 1, 2, \dots \right\}.
\]
Let $\mathcal{A} = \{A_1, A_2, \dots, A_n\}$ and $U(\mathcal{A})$ be the multiset formed from the disjoint union of the sequences. The problem is to find the value of the $k$th smallest element of $U(\mathcal{A})$.

We do not actually produce the $k$ smallest elements, only the value of the $k$th smallest one, allowing us to eliminate any dependence on $k$ in our time bounds. An explicit list of the $k$ smallest elements is also not necessary for the election problem, since we are primarily interested in the total amount of resources allocated to each entity, which can be calculated from the value of the $k$th smallest element. When a rank function (defined in the following paragraph) can be computed in constant time, we may use it to compute, in constant time for each $A_i$, the largest index $j$ such that $d_j/v_i$ is at most the computed value, which gives the allocation to entity~$i$. Producing only the value of the $k$th smallest element also sidesteps the issue of tie-breaking when several entities have the same priorities and are equally eligible for the last resource. The rules for breaking ties are application-dependent, so it is best to leave them out of the main algorithm.

For a sequence $A$, let $A(j)$ denote the $j$th element of the sequence $A$, with the first element at index $0$. We let $A(-1) = -\infty$ to avoid corner cases in the algorithm; however, when counting elements of $A$, this $-\infty$ value should be ignored.
Define the \emph{rank} of a number $x$ in $A$ as the number of elements of $A$ less than or equal to $x$. Equivalently, for a strictly monotonic sequence, this is the index $j$ such that $A(j) \le x < A(j+1)$. We denote the rank function by $r(x, A)$. For a set $\mathcal{A}$ of sequences, we define $r(x, \mathcal{A}) = \sum_{A \in \mathcal{A}} r(x, A)$. If the set $\mathcal{A}$ is clear from context, we will drop $\mathcal{A}$ and simply write $r(x)$.

In our algorithm, we assume that the rank function for each sequence $A_i$ can be computed in constant time. When the sequences are arithmetic progressions (as they are in most of the voting methods we consider), these functions can be computed using only a constant number of basic arithmetic operations, so this is not a restrictive assumption. The Huntington--Hill method involves square roots, but its rank function may still be calculated using a constant number of operations that are standard enough to be included as hardware instructions on modern processors.

Observe a small subtlety about the rank function: if $\tau$ is the $k$th smallest element, then $r(\tau)$ is not necessarily $k$. Indeed, $r(\tau)$ can be greater than $k$, as in the case where there are $k-1$ elements of $U(\mathcal{A})$ less than $\tau$ and $\tau$ is duplicated twice, in which case $r(\tau) = k + 1$. In general, we have $k \le r(\tau) \le k + n - 1$. The rank of the $k$th smallest element can still be characterized, through the following observation.
\begin{observation} \label{obs:tau-def}
$\tau$ is the value of the $k$th smallest element in $U(\mathcal{A})$ if and only if $r(\tau) \ge k$ and for all $x < \tau$, $r(x) < k$.
\end{observation}

Now define $L(x, A)$ as the largest value in $A$ less than $x$; similary, define $G(x, A)$ as the smallest value in $A$ greater than $x$. Note that $L(x, A)$ and $G(x, A)$ can be computed easily from the rank of $x$: if $r = r(x, A)$, then $L(x, A)$ and $G(x, A)$ must be the value of either $A(r-1)$, $A(r)$, or $A(r+1)$. For $\mathcal{A}$, we use the similar notation $L(x, \mathcal{A})$ (respectivley $G(x, \mathcal{A})$) to denote the multiset of $L(x, A)$ (respectively $G(x, A)$) over all $A \in \mathcal{A}$.

Lastly, we make note of one notational convention. In our descriptions, the input variables to an algorithm may change within the algorithm, and it is useful to talk about both the values of the variables as they change and their initial values. Therefore, we use a tilde to denote the initial value of a variable, and the lack of a tilde to denote the changing value over the course of the algorithm. For example, $\tilde{\mathcal{A}}$ means the initial value and $\mathcal{A}$ means the value at an intermediate point of the algorithm.

\section{The Algorithm}

Our algorithm has three parts. In the first part, we show how the value of the $k$th smallest element of $U(\mathcal{A})$ can be found from a \emph{coarse solution}, a value whose rank is within $O(n)$ positions of~$k$. In the second part, we handle a special case of the problem in which every sequence in $\mathcal{A}$ is an arithmetic sequence, by showing how the rank function over $\mathcal{A}$ can be inverted in this case to produce a coarse solution. And in the last part, we deal with more general sequences, by showing how arithmetic sequences that approximate them can be used to produce a coarse solution.

\subsection{From Coarse to Exact Solutions}

In this section, we show how to compute the value of the $k$th smallest element of $U(\mathcal{A})$, given a coarse solution. A value $\xi$ is called a \emph{coarse solution} for $k$ if $|k - r(\xi)| \le cn$ for some constant $c$. Equivalently, this means there are only $O(n)$ elements between $\xi$ and the $k$th smallest element. Note that $\xi$ does not have to be an element of $U(\mathcal{A})$. 

Before presenting the algorithm, we first show that the coarse solution can be assumed to have a rank smaller than $k$.

\begin{algorithm}[tb]
\caption{\textsc{LowerRankCoarseSolution}($\tilde{\mathcal{A}}$, $\tilde{k}$, $\xi$)}
\label{alg:lower-rank-coarse-solution}
\begin{algorithmic}[1]
  \REQUIRE $\tilde{\mathcal{A}}$: set of increasing sequences; $\tilde{k}$: positive integer; $\xi$: a coarse solution with $r(\xi, \tilde{\mathcal{A}}) \ge k$
  \ENSURE another coarse solution $\xi'$ with $r(\xi', \tilde{\mathcal{A}}) < k$

  \STATE $\mathcal{A} \leftarrow \tilde{\mathcal{A}}$, $k \leftarrow \tilde{k}$, $u \leftarrow \xi$
  \LOOP
    \STATE $\bar{x} \leftarrow \text{median of } L(u, \mathcal{A})$
    \IF{$r(\bar{x}, \mathcal{A}) \ge k$}
      \STATE $u \leftarrow \bar{x}$
    \ELSIF{$r(\bar{x}, \mathcal{A}) < k - n$}
      \STATE $\mathcal{B} \leftarrow \{A : A \in \mathcal{A} \text{ and } L(u, A) \le \bar{x} \}$
      \STATE $k \leftarrow k - \sum_{A \in \mathcal{B}} r(\bar{x}, A)$
      \STATE $\mathcal{A} \leftarrow \mathcal{A} \setminus \mathcal{B}$
    \ELSE
      \RETURN $\bar{x}$
    \ENDIF
  \ENDLOOP
\end{algorithmic}
\end{algorithm}

\begin{lemma} \label{lem:lower-rank-coarse-solution}
Let $\xi$ be a coarse solution for $k$, and assume $r(x, A)$ can be computed in constant time for every $A \in \mathcal{A}$. If $r(\xi, \mathcal{A}) \ge k$, then another coarse solution $\xi'$ with $r(\xi', \mathcal{A}) < k$ can be found in $O(n)$ time.
\end{lemma}

\begin{proof}
We find a $\xi'$ such that $\tilde{k} - n \le r(\xi', \tilde{\mathcal{A}}) < \tilde{k}$. To find this value, start with $u = \xi$. Then, repeatedly update $u$ and $\mathcal{A}$ as follows, until the median $\bar{x}$ of $L(u, \mathcal{A})$ has rank between $k - n$ and $k$:
\begin{enumerate}
\item If $r(\bar{x}, \mathcal{A}) \ge k$, then set $u = \bar{x}$.
\item If $r(\bar{x}, \mathcal{A}) < k - n$, then any sequence $A$ in $\mathcal{A}$ with $L(u,A) \le \bar{x}$ can no longer help us get closer to a value in the desired range, so we remove those sequences and update $k$ accordingly to compensate for their removal (i.e., subtract from $k$ the ranks $r(\bar{x}, A)$ over all removed sequences $A$).
\end{enumerate}
Algorithm~\ref{alg:lower-rank-coarse-solution} summarizes this procedure.

Let $q$ be the sum of two quantities: the distance from the rank of $u$ to $\tilde k$, and the number of sequences remaining in $\mathcal{A}$. Then $q$ is initially $O(n)$ by the assumption that $\tilde\xi$ is a coarse solution. Each iteration of the loop of the algorithm takes time $O(|\mathcal{A}|)$ and reduces $q$ by $O(|\mathcal{A}|)$ units, either by reducing the rank of $u$ in the first case or by eliminating sequences from $\mathcal{A}$ in the second case. Eventually (before $q$ can be reduced to zero) the algorithm must terminate, at which point it has taken time proportional to the total reduction in $q$, which is $O(n)$. When it terminates, the returned value is clearly a coarse solution whose rank is less than~$k$, as desired.
\end{proof}

We now present the algorithm to convert a coarse solution to an exact one. The algorithm is similar to a binary search, where we maintain both a lower bound and an upper bound that narrow the possible candidates as the algorithm progresses. The lower bound is initially derived from the coarse solution, which guarantees that it is close to the true solution. The main difference between our algorithm and a standard binary search is that we do not know the distribution of the sequences' elements within the bounds, so we cannot reduce search space by a constant proportion simply by splitting the range halfway between the lower and upper bounds. Instead, we split on the median of some well-chosen set within the bounds, which will allow us to reduce the number of candidates by at least $|\mathcal{A}| / 2$ in each step. To make the algorithm run in linear time, sequences that no longer have elements between the bounds are removed from $\mathcal{A}$. The procedure is similar to the one in Algorithm~\ref{alg:lower-rank-coarse-solution}. But instead of moving down in $U(\mathcal{A})$ with $L(\cdot, \mathcal{A})$, the algorithm moves up with $G(\cdot, \mathcal{A})$. Because of this, compensating $k$ when removing a sequence is no longer as straight-forward. In particular, we may need to query the rank of the upper bound $u$, so we need an extra variable to keep track of the possible under-compensation to the rank for these values. The details are presented in the theorem below.

\begin{theorem}\label{thm:coarse-to-true-solution}
Let $\mathcal{A}$ be a set of increasing sequences. Assume $r(x, A)$ can be computed in constant time for every $A \in \mathcal{A}$. If a coarse solution $\xi$ is given, then the value of the $k$th smallest element in $U(\mathcal{A})$ can be found in $O(n)$ time.
\end{theorem}

For space reasons we defer a detailed proof to Appendix~\ref{sec:proof}.

\begin{algorithm}[t]
\caption{\textsc{CoarseToExact}($\tilde{\mathcal{A}}$, $\tilde{k}$, $\xi$)}
\label{alg:coarse-to-true-solution}
\begin{algorithmic}[1]
  \REQUIRE $\tilde{\mathcal{A}}$: set of increasing sequences; $\tilde{k}$: positive integer; $\xi$: coarse solution with $r(\xi, \tilde{\mathcal{A}}) < k$
  \ENSURE the value of the $\tilde{k}$th smallest element in $U(\mathcal{\tilde{A}})$

  \STATE $\mathcal{A} \leftarrow \tilde{\mathcal{A}}$, $k \leftarrow \tilde{k}$
  \STATE $l \leftarrow \xi$, $u \leftarrow \infty$, $m \leftarrow 0$
  \REPEAT
    \STATE $\bar{x} \leftarrow \text{median of } G(l, \mathcal{A})$
    \IF{$r(\bar{x}, \mathcal{A}) < k$}
      \STATE $l \leftarrow \bar{x}$
    \ELSE
      \STATE $u \leftarrow \bar{x}$
      \STATE $m \leftarrow 0$
    \ENDIF

    \IF{$| \{A : G(l, A) = u \} | \ge 1$}
      \STATE $m \leftarrow m + | \{A : G(l, A) = u \} | - 1$
    \ENDIF
    \STATE $\mathcal{A}' \leftarrow \{A : A \in \mathcal{A} \text{ and } G(l, A) < u \} \cup \{ \text{any one } A \in \mathcal{A} \text{ such that } G(l, A) = u \}$
    \STATE $k \leftarrow k - \sum_{A \in \mathcal{A} \setminus \mathcal{A}'} r(l, A)$
    \STATE $\mathcal{A} \leftarrow \mathcal{A}'$
  \UNTIL{$G(l, \mathcal{A})$ has only one value $t$ \AND ( $r(t, \mathcal{A}) \ge k$ \OR ($t = u$ \AND $r(t, \mathcal{A}) \ge k - m$) )}

  \RETURN $t$
\end{algorithmic}
\end{algorithm}

\subsection{Coarse Solution for Arithmetic Sequences}

In this section, we focus on the special case where every sequence in $\mathcal{A}$ is an arithmetic sequence. In particular, each sequence $A$ is of the form $A(j) = x_A + y_A \cdot j$ with $y_A > 0$, for $j = 0, 1, 2, \dots$. In this case, the rank function is given by
\begin{align}\label{eq:r(x)-ap}
  r(x) = \sum_{A \in \mathcal{A}} r(x, A)
       = \sum_{A \in \mathcal{A}} \left(1 + \left\lfloor \frac{x - x_A}{y_A} \right\rfloor \right) I_{x \ge x_A}
\end{align}
where $I_{x \ge x_A}$ is the indicator function that is $1$ when $x \ge x_A$ and $0$ otherwise.

We can think of the problem of finding the value of the $k$th smallest element as ``inverting'' the rank function to find $x$ such that $r(x)$ is close to $k$. Of course, the inverse of $r(\cdot)$ does not make sense, since the function is neither one-to-one nor onto $\mathbb{N}$. However, if we drop the floor and the plus one, and consider
\begin{align}\label{eq:r(x)-aps}
  s(x) = s(x, \mathcal{A}) = \sum_{A \in \mathcal{A}} s(x, A) = \sum_{A \in \mathcal{A}} \left( \frac{x - x_A}{y_A} \right) I_{x \ge x_A},
\end{align}
the resulting function is piecewise linear with a well-defined inverse.

Call the sequence $A$ \emph{contributing} for $k$ if $s(A(0)) = s(x_A) \le k$. To invert $s(x)$, we need to first find the contributing sequences of $\mathcal{A}$ for $k$. This can be done in $O(n)$ time, as shown in the next lemma.

\begin{lemma}\label{lem:contrib-seqs}
Suppose $\mathcal{A}$ is a set of arithmetic sequences of the form $A(j) = x_A + y_A \cdot j$ with $y_A > 0$, then the contributing sequences of $\mathcal{A}$ can be found in $O(n)$ time.
\end{lemma}

\begin{proof}
The algorithm to find the contributing sequences is given in Algorithm~\ref{alg:find-contributing-seqs}. In the main loop, we repeatedly compute the median $\bar{x}$ of the first element of remaining sequences $\mathcal{A}$. Comparison of $s(\bar{x}, \tilde{\mathcal{A}})$ to $k$ then eliminates a portion of those sequences and determines which of the eliminated sequences are contributing:
\begin{enumerate}
\item If $s(\bar{x}, \tilde{\mathcal{A}}) > k$: None of the sequences with $x_A \ge \bar{x}$ can be contributing, so we reduce $\mathcal{A}$ to only those with $x_A < \bar{x}$.
\item If $s(\bar{x}, \tilde{\mathcal{A}}) \le k$: The sequences with $x_A \le \bar{x}$ are contributing, so we add them all to the list $\mathcal{C}$, and reduce $\mathcal{A}$ to only those sequences with $x_A > \bar{x}$.
\end{enumerate}
The loop repeats until $\mathcal{A}$ is empty, at which time $\mathcal{C}$ is the output.

The algorithm's correctness follows from the fact that a sequence is added to $\mathcal{C}$ if and only if it is contributing. For the algorithm to run in $O(n)$ time, each iteration of the loop must run in $O(|\mathcal{A}|)$ time. In particular, we must compute $s(\bar{x}, \tilde{\mathcal{A}}) = s(\bar{x}, \mathcal{A}) + s(\bar{x}, \mathcal{C})$ in $O(|\mathcal{A}|)$ time. To do this, note that when a sequence $A$ is added to $\mathcal{C}$, all subsequent $\bar{x}$ have $\bar{x} > x_A$. This means we can consider $s(\cdot, \mathcal{C})$ as a linear function, whose value can be computed in constant time by keeping track of the two coefficients of the linear function and updating them whenever sequences are added to $\mathcal{C}$.
\end{proof}

\begin{algorithm}[tb]
\caption{\textsc{FindContributingSequences}($\tilde{\mathcal{A}}$, $k$)}
\label{alg:find-contributing-seqs}
\begin{algorithmic}[1]
  \REQUIRE $\tilde{\mathcal{A}}$: set of arithmetic sequences of the form $A(j) = x_A + y_A j$; $k$: integer $> n$
  \ENSURE subset $\mathcal{C}$ of $\tilde{\mathcal{A}}$ of contributing sequences for $k$
  \STATE $\mathcal{A} \leftarrow \tilde{\mathcal{A}}$, $\mathcal{C} \leftarrow \emptyset$.
  \WHILE{$\mathcal{A}$ is nonempty}
    \STATE $\bar{x} \leftarrow \text{median of } \{x_A : A \in \mathcal{A} \}$
    \IF{$s(\bar{x}, \tilde{\mathcal{A}}) > k$}
      \STATE $\mathcal{A} \leftarrow \{A : A \in \mathcal{A} \text{ and } x_A < \bar{x}\}$
    \ELSE
      \STATE $\mathcal{C} \leftarrow \mathcal{C} \cup \{A : A \in \mathcal{A} \text{ and } x_A \le \bar{x} \}$
      \STATE $\mathcal{A} \leftarrow \{A : A \in \mathcal{A} \text{ and } x_A > \bar{x}\}$
    \ENDIF
  \ENDWHILE
  \RETURN $\mathcal{C}$
\end{algorithmic}
\end{algorithm}

Once the contributing sequences $\mathcal{C}$ are found, we can restrict our search for $x$ to the last interval of $s(\cdot, \mathcal{C})$, which is a linear function without breakpoints and can be inverted easily. In particular, the inverse is given by
\begin{align}\label{eq:s_inv(k)}
  s^{-1}(k) = \left( \sum_{A \in \mathcal{C}} \frac{1}{y_A} \right)^{-1} \left( k + \sum_{A \in \mathcal{C}} \frac{x_A}{y_A} \right)
\end{align}
which can be interpreted as a weighted and shifted harmonic mean of the slopes of contributing arithmetic sequences. Since $r(x) - s(x) \le n$, $s^{-1}(k)$ is a coarse solution, applying \autoref{thm:coarse-to-true-solution} immediately gives the following.

\begin{theorem}\label{thm:ap}
If $\mathcal{A}$ is a set of arithmetic sequences, then the value of the $k$th smallest element in $U(\mathcal{A})$ can be found in $O(n)$ time.
\end{theorem}

While not necessary for subsequent sections, a slight generalization can be made here. The main property of arithmetic sequences we used is that $r(x, A)$ is approximable by a function that can be written as $w_A f(x) + u_A$, where $f$ is an invertible function independent of $A$; $w_A, u_A$ are constants; and $w_A > 0$. When $r(x,A)$ can be approximated in this way, the inverse is given by:
\begin{align}
  s^{-1}(k) = f^{-1}\left( \frac{k - \sum_{A \in \mathcal{C}} u_A}{\sum_{A \in \mathcal{C}} w_A} \right)
\end{align}
where $\mathcal{C}$ is the set of contributing sequences.
For example, if $A(j) = 2^j / v_A$, then the rank function is given by $r(x, A) = \lfloor \log (x v_A) \rfloor$ and is approximable by $\log x + \log v_A$. The coarse solution is then given by
\begin{align} \label{eq:2^j-approx}
  s^{-1}(k) = e^{\frac{1}{|\mathcal{C}|} \left( k - \sum_{A \in \mathcal{C}} \log v_A \right)}.
\end{align}

\subsection{Coarse Solution for Approximately-Arithmetic Sequences}

In this section, we show how to handle more general sequences for highest averages methods. Our strategy works as long as the divisor sequence $D = \{d_j\}$ is close to an arithmetic progression $E = \{e_j\}$, where closeness here means that there is a constant $c$ so that $|d_j - e_j| \le c$ for every $j \ge 0$. 

Suppose $\mathcal{A}$ is the set of (arithmetic) sequences induced by the arithmetic divisor sequence $E$ and $\mathcal{B}$ is the set of sequences induced by the divisor sequence $D$. We show that, if $E$ and $D$ are close, then the rank of every number $x$ in $\mathcal{A}$ is within a constant of the rank in $\mathcal{B}$. This means that a coarse solution for $\mathcal{A}$ is also a coarse solution for $\mathcal{B}$. Hence, to do apportionment for these more general sequences, we just use the results from the previous section to find a coarse solution for the approximating arithmetic sequences, then apply \autoref{thm:coarse-to-true-solution} to the original sequences with that coarse solution.

\begin{lemma}\label{lem:arithm-sequence-approx}
Let $A$ be an arithmetic progression $A(j) = x_A + y_A j$ with slope $y_A > 0$, and let $B$ be an increasing sequence such that for every $j$, $|A(j) - B(j)| \le c y_A$ for some constant $c$. Then $|r(x, A) - r(x, B)| \le c'$ for another constant $c'$.
\end{lemma}

\begin{proof}
Fix $x$. Let $a = r(x, A)$ and $b = r(x, B)$. Note that $A(a)$ (respectively $B(b)$) is largest value of $A$ (respectively $B$) less than or equal to $x$. We have two cases:
\begin{description}
\item[Case $A(b) \ge x$:] Note that $a \le b$ since $A(a) \le x \le A(b)$. Furthermore,
\[
  A(b) - A(a) \le [A(b) - B(b)] + [x - A(a)] \le c y_A + y_A.
\]
This means that there are at most $c+1$ elements of $A$ between $A(a)$ and $A(b)$, so $b - a \le c + 1$.

\item[Case $A(b) < x$:] Note that $b \le a$. Furthermore, 
\begin{align*}
  x - B(b) &  <  B(b+1) - B(b) \\
           & \le |B(b+1) - A(b+1)| + |A(b+1) - A(b)| + |A(b) - B(b)| \\
           & \le c y_A + y_A + c y_A = (2c + 1) y_A
\end{align*}
It follows that
$
  A(a) - A(b) \le [x - B(b)] + |B(b) - A(b)| \le (2c + 1) y_A + c y_A,
$
and $a - b \le 3c+1$.
\end{description}
In both cases, we have $|r(x, A) - r(x, B)| = |b - a| \le \text{const}$, as needed.
\end{proof}

From the lemma, the following is immediate.

\begin{theorem}\label{thm:divisors-approx}
Let the divisor sequence $E=\{e_j\}$ be an arithmetic progression, and suppose $D=\{d_j\}$ is another divisor sequence such that for every $j \ge 0$, $|d_j - e_j| \le c$ for some constant $c$. If $\mathcal{A}$ is the set of sequences induced by $E$ and $\mathcal{B}$ the set of sequences induced by $D$, then for every $x$, $|r(x, \mathcal{A}) - r(x, \mathcal{B})| \le c'n$, for some constant $c'$.
\end{theorem}

\begin{proof}
Let $A \in \mathcal{A}$ and $B$ be the corresponding sequence in $\mathcal{B}$ with score $v$. We have
\[
  |A(j) - B(j)| = \frac{|d_j - e_j|}{v} \le \frac{c}{v}.
\]
Note that the slope of $A$ is $1/v$, so by \autoref{lem:arithm-sequence-approx}, $|r(x, A) - r(x, B)| \le c'$ for some constant $c'$. Summing over all sequences gives the desired result.
\end{proof}

The strategy in this section works for all the methods in \autoref{tab:divisors} with non-arithmetic divisor sequences. The Huntington--Hill and Dean methods have divisor sequences where each element of the sequence is the geometric or harmonic mean respectively of consecutive natural numbers. Hence, each element is within $1$ of the arithmetic sequence $d_j = j$, so finding a coarse solution under the divisor sequence $d_j = j$ is enough to find a coarse solution for those two methods.
For methods like the modified Sainte-Lagu\"e method, where the divisor sequence is arithmetic except for a constant number of elements, we may ignore the non-arithmetic elements and find a coarse solution solely from the remaining elements that do form an arithmetic progression.

\section{Conclusion}
We have shown that many commonly used apportionment methods can be implemented in time linear in the size of the input (vote totals for each entity), based on a transformation of the problem into selection in multisets formed from unions of arithmetic or near-arithmetic sequences.
Our method can be extended to selection in unions of other types of sequences, as long as the rank functions of the sequences can be approximately inverted and aggregated.
It would be of interest to determine whether  forms of diophantine approximation from other application areas can be computed as efficiently. 

\subsection*{Acknowledgements}
This research was supported in part by NSF grant 1228639, and
ONR grant N00014-08-1-1015

{\raggedright
\bibliographystyle{splncs}
\bibliography{paper}}

\newpage
\appendix

\section{Proof of \autoref{thm:coarse-to-true-solution}}
\label{sec:proof}

We present an algorithm that computes the value of the $k$th smallest element of $U(\mathcal{A})$ in $O(n)$ time (see Algorithm~\ref{alg:coarse-to-true-solution}). By \autoref{lem:lower-rank-coarse-solution}, we may assume $r(\xi, \tilde{\mathcal{A}}) < k$. 

The algorithm starts by initializing the upper bound $u$ to $\infty$ and the lower bound $l$ to $\xi$. Next, it initializes a variable $m$ to be zero. This is the variable that will compensate for the under-adjustment to $k$ at the upper bound $u$ when certain sequences are removed. After that, the main loop of the algorithm starts, which repeats until $G(l, \mathcal{A})$ has only one value among its elements and that value has rank $\ge k$ (or $\ge k - m$ if that value is $u$). This condition guarantees that $l$ is the tightest lower bound to the desired value, in the sense that there can be no other value in $U(\mathcal{A})$ between $l$ and the value of the $k$th smallest element. Let $\bar{x} = \text{median of } G(l, \mathcal{A})$. In the loop, we have two cases:
\begin{enumerate}
\item Case 1: $r(\bar{x}, \mathcal{A}) < k$. In this case, the value of the $k$th smallest element is to the right of $\bar{x}$, so we adjust the lower bound $l$ to $\bar{x}$.
\item Case 2: $r(\bar{x}, \mathcal{A}) \ge k$. In this case, the value of the $k$th smallest element is to the left of $\bar{x}$, so we adjust the upper bound $u$ to $\bar{x}$. Since $u$ is updated, we also reset $m$ to be zero.
\end{enumerate}
Lastly, we need to remove the sequences in $\mathcal{A}$ that are no longer useful. In particular, we remove all sequences $A$ in $\mathcal{A}$ with $G(l, A) > u$ and all but one (arbitrarily chosen) sequence with $G(l, A) = u$. We need to remove the sequences with $G(l, A) = u$ so that we can guarantee enough values are eliminated in this step, but we cannot remove all of them since $u$ might actually be the desired value. We increment $m$ by the number of these sequences with $G(l, A) = u$ that are removed. Then we update $k$ by subtrating the contribution from the removed sequences; namely, we set $k = k - \sum_{A \in \mathcal{A} \setminus \mathcal{A}'} r(l, A)$ where $\mathcal{A} \setminus \mathcal{A}'$ is the set of removed sequences. When the main loop ends, $G(l, \mathcal{A})$ should contain only a single value; that value is the output of the algorithm.

\paragraph{Correctness.} Let $\mathcal{B} = \tilde{\mathcal{A}} \setminus \mathcal{A}$. For each $A \in \mathcal{B}$, let $l_A$ be the value of $l$ when the sequence $A$ was removed. To prove that the algorithm works, we show that the main loop satisfies the following invariant conditions:
\begin{enumerate}
\item $\mathcal{A} \neq \emptyset$
\item $r(l, \mathcal{A}) < k$
\item $\tilde{k} = k + \sum_{A \in \mathcal{B}} r(l_A, A)$
\item for all $x \in [l, u)$, $r(x, \mathcal{B}) = \sum_{A \in \mathcal{B}} r(l_A, A)$.
\item $r(u, \mathcal{B}) = m + \sum_{A \in \mathcal{B}} r(l_A, A)$
\end{enumerate}
Initially, condition (2) is true because $l = \xi$, and the other conditions are true because $\tilde{\mathcal{A}} = \mathcal{A}$, $\tilde{k} = k$, and no sequences have been removed. At the end of each iteration of the loop, conditions (1) to (3) are true from the way $\mathcal{A}$, $l$, and $k$ are updated, and conditions (4) and (5) are true because $l$ only increases, $u$ only decreases, and a sequence is removed only when it has no elements between $l$ and $u$.

Consider when the loop terminates. By condition (1), the return value makes sense. By conditions (3) to (5) and the terminating condition, we have
\begin{align} \label{eq:r(t)>=k}
  r(t, \tilde{\mathcal{A}}) =   r(t, \mathcal{A}) + r(t, \mathcal{B})
                            \ge k - m + m + \sum_{A \in \mathcal{B}} r(l_A, A)
                            =   \tilde{k}
\end{align}
Similarly, by conditions (2) to (4), we have
\begin{align} \label{eq:r(l)<k}
  r(l, \tilde{\mathcal{A}}) = r(l, \mathcal{A}) + r(l, \mathcal{B})
                            < k + \sum_{A \in \mathcal{B}} r(l_A, A)
                            = \tilde{k}
\end{align}
By the terminating condition and the fact that removed sequences have no values between $l$ and $u$, $U(\tilde{\mathcal{A}})$ has no elements between $l$ and $t$. Applying Observation~\ref{obs:tau-def}, the output $t$ must be the value of the $k$th smallest element.

\paragraph{Run Time.} We now show that the algorithm terminates in $O(n)$ time. Let $\tau$ be the value of the $k$th smallest element. After every iteration, either $l$ gets closer to $\tau$ by at least $\lfloor |\mathcal{A}| / 2 \rfloor$, or the number of remaining sequences reduces by at least $\lfloor |\mathcal{A}| / 2 \rfloor$.
At some point, $l$ must be the tightest lower bound to $\tau$. When that happens, every subsequent iteration satisfies the second case and reduces the size of $\mathcal{A}$ until the terminating condition is true. This shows that the algorithm terminates. 
Now, each iteration takes $O(|\mathcal{A}|)$ time, so as in \autoref{lem:lower-rank-coarse-solution}, each iteration makes progress proportional to the amount of work. The distance between $l$ and $\tau$ is $O(n)$ and the number of sequences is $n$, so the total running time is $O(n)$.
\qed

\end{document}